\documentclass[letterpaper, 10 pt, conference]{ieeeconf}
\IEEEoverridecommandlockouts
\let\proof\relax
\let\endproof\relax
\let\labelindent\relax
\usepackage{enumitem, amsmath, amsthm, amssymb, amsfonts, mathtools, xargs, tensor, units, cite, stmaryrd, mathrsfs, algpseudocode, comment, graphicx}
\usepackage[unicode=true, bookmarks=true, bookmarksnumbered=false, bookmarksopen=false, breaklinks=false, pdfborder={0 0 1}, backref=false, colorlinks=false, hidelinks]{hyperref}

\newcommand{\D}{\mathrm{d}}

\newcommand{\R}{\mathbb{R}}

\newcommand{\B}{\operatorname{B}}

\newcommand{\N}{\mathbb{N}}

\newcommandx\fid[5][usedefault, addprefix=\global, 1=, 2=, 3=, 4=, 5=]{\tensor*[_{#3}^{#2}]{#1}{_{#5}^{#4}}}

\newcommand{\g}{\operatorname{g}}

\newcommand{\id}{\operatorname{I}}

\usepackage[bytype]{Theorems}
\begin{document}
\title{Output-feedback online optimal control for a class of nonlinear systems}
\author{Ryan Self, Michael Harlan, and Rushikesh Kamalapurkar\thanks{The authors are with the School of Mechanical and Aerospace Engineering, Oklahoma State University, Stillwater, OK, USA. {\tt\small \{rself, michael.c.harlan, rushikesh.kamalapurkar\}@okstate.edu}.}}
\maketitle
\begin{abstract}
	In this paper an output-feedback model-based reinforcement learning (MBRL) method for a class of second-order nonlinear systems is developed. The control technique uses exact model knowledge and integrates a dynamic state estimator within the model-based reinforcement learning framework to achieve output-feedback MBRL. Simulation results demonstrate the efficacy of the developed method.
\end{abstract}
\section{Introduction\label{sec:Introduction}}
Over the past decade, online reinforcement learning algorithms that guarantee stability during the learning phase have been developed for deterministic systems \cite{SCC.Chen.Jagannathan2008,SCC.Mehta.Meyn2009,SCC.Vrabie.Lewis2010,SCC.Vamvoudakis.Lewis2010,SCC.Lewis.Vrabie.ea2012,SCC.Lee.Park.ea2012,SCC.Modares.Lewis.ea2013,SCC.Bhasin.Kamalapurkar.ea2013a,SCC.Bian.Jiang.ea2014,SCC.Kiumarsi.Lewis.ea2014,SCC.Modares.Lewis2014,SCC.Bian.Jiang.ea2015,SCC.Li.Liu.ea2015,SCC.Yang.Liu.ea2015,SCC.Zhao.Xu.ea2015}; however, stability and convergence are established under restrictive persistence of excitation (PE) conditions which are difficult, if not impossible, to verify. To soften the PE condition, data-driven methods that employ experience replay have been utilized in results such as \cite{SCC.Cichosz1999,SCC.Wawrzynski2009,SCC.Zhang.Cui.ea2011,SCC.Adam.Busoniu.ea2012,SCC.Luo.Wu.ea2014,SCC.Modares.Lewis.ea2014,SCC.Li.Liu.ea2015}; however, since the data is collected along the system trajectory, added exploration signals are often required to achieve convergence. The need for PE and exploration signals is a result of sample inefficiency, and is a significant drawback of the existing model-free RL-based online optimal control methods.

Model-based reinforcement learning (MBRL) algorithms learn a model of the system from observations using supervised learning and employ the model to learn the policies. Several different MBRL approaches have been developed in the literature over the last few decades. Imaginary roll-outs, i.e., the use of a model as a proxy for the real world to evaluate temporal difference errors (referred to as Bellman errors (BEs) in this paper) are explored in results such as \cite{SCC.Sutton1991a} and  \cite{SCC.Lampe.Riedmiller2014}. While the sample efficiency is of the policy learning algorithms is improved, the performance of the method in \cite{SCC.Sutton1991a} decays rapidly with model mismatch, and the method in \cite{SCC.Lampe.Riedmiller2014} relies on fitting neural networks to dynamics, which is typically data-intensive, nullifying the sample efficiency gain in the policy learning algorithm.

Policy gradient methods that rely on backpropagation through the model to compute the gradient of the state or the action value function with respect to the policy parameters are developed in results such as \cite{SCC.Abbeel.Quigley.ea2006,SCC.Deisenroth.Rasmussen2011,SCC.Heess.Wayne.ea2015,SCC.Grondman2015}; however, policy gradient methods are often iterative in nature and typically do not study stability during the learning phase, and as a result, are not suitable for real-time simultaneous learning and execution. MBRL methods with provable sample efficiency bounds have been developed in results such as \cite{SCC.Brafman.Tennenholtz2002,SCC.Kearns.Singh2002,SCC.Kakade.Kearns.ea2003,SCC.Nouri.Littman2009,SCC.Li.Littman.ea2011,SCC.Jung.Stone2010,SCC.Grande.Walsh.ea2014}; however, the theoretical guarantees are obtained under discretization of the continuous state space into finitely many discrete states and a finite action space, and as such, are not directly applicable to systems with continuous state and action spaces.

The MBRL technique developed by the authors in \cite{SCC.Kamalapurkar.Dinh.ea2015,SCC.Kamalapurkar.Walters.ea2016,SCC.Kamalapurkar.Walters.ea2016a,SCC.Kamalapurkar.Rosenfeld.ea2016,SCC.Kamalapurkar.Andrews.ea2017,SCC.Kamalapurkar.Klotz.ea2018} for continuous time and continuous space systems softens the excitation requirements used in results such as \cite{SCC.Chen.Jagannathan2008,SCC.Mehta.Meyn2009,SCC.Lewis.Vrabie2009,SCC.Vrabie.Lewis2009,SCC.Vrabie.Lewis2010,SCC.Vamvoudakis.Lewis2010,SCC.Lewis.Vrabie.ea2012,SCC.Vamvoudakis.Lewis2011a,SCC.Vamvoudakis.Lewis.ea2012,SCC.Lee.Park.ea2012,SCC.Modares.Lewis.ea2013,SCC.Bhasin.Kamalapurkar.ea2013a,SCC.Bian.Jiang.ea2014,SCC.Kiumarsi.Lewis.ea2014,SCC.Modares.Lewis2014,SCC.Bian.Jiang.ea2015,SCC.Song.Lewis.ea2015,SCC.Li.Liu.ea2015,SCC.Yang.Liu.ea2015,SCC.Zhao.Xu.ea2015} by utilizing a model of the system to simulate exploration, where the stability and the performance of the closed-loop system critically depends on the accuracy of the estimated model. A significant drawback of the online optimal control methods mentioned so far is that they require full state measurements. 

While model-based and model-free reinforcement learning can be achieved using output feedback instead of state feedback by making use of partially observable Markov decision processes (POMDPs), in general, POMDPs are undecidable if the objective is to find an optimal solution, and finding a near-optimal solution can also be NP-hard \cite{SCC.Papadimitriou.Tsitsiklis1987,SCC.Madani.Hanks.ea2003}. In this paper, the problem is formulated as a state estimation based reinforcement learning problem, and for a specific class of systems, an online solution is obtained that guarantees stability during the learning phase.

A recent result in \cite{SCC.Modares.Lewis.ea2016}, presents an offline model-free algorithm for linear systems to achieve optimality using output feedback. The objective in this paper is to develop an output-feedback model-based reinforcement learning method for a class of nonlinear systems under exact model knowledge. While the developed results can be extended to systems with uncertain models using model-learning methods such as \cite{SCC.Kamalapurkar2017a}, such extension is not a focus of this work.

The paper is organized as follows. A detailed description of the problem under consideration is provided in Section \ref{sec:Introduction}. To facilitate the subsequent analysis of the developed technique, section \ref{sec:Stability Under Optimal Feedback} examines the stability properties of optimal controllers under semidefinite cost functions for the class of systems under consideration. Section \ref{sec:Velocity-estimator-design} describes the state estimator used in the design. Section \ref{sec:Model-based Reinforcement Learning} describes the developed MBRL method. Section \ref{sec:Analysis} presents a Lyapunov-based stability analysis, Section \ref{sec:Simulation Results} presents simulation results, and Section \ref{sec:Conclusion} concludes the paper.
\section{Problem Description\label{sec:Problem Description}}
Consider a second order nonlinear system of the form
\begin{align}
\dot{p} & =q,\nonumber \\
\dot{q} & =f\left(x\right)+g\left(x\right)u,\nonumber \\
y & =p,\label{eq:Dynamics}
\end{align}
where $p\in\R^{n}$ and $q\in\R^{n}$ denote the generalized position states and the generalized velocity states, respectively, $x\coloneqq\begin{bmatrix}p^{T} & q^{T}\end{bmatrix}^{T}$ is the system state, $f:\R^{2n}\to\R^{n}$ is locally Lipschitz continuous, $ f\left(0\right)=0 $, and $y\in\R^{n}$ denotes the output. The drift dynamics $f$ are unknown and the control effectiveness $g:\R^{2n}\to\R^{n\times m}$ is known and locally Lipschitz. Systems of the form (\ref{eq:Dynamics}) encompass second-order linear systems and Euler-Lagrange models with known inertial matrices, and hence, represent a wide class of physical plants, including, but not limited to, robotic manipulators and autonomous ground, aerial, and underwater vehicles.

The objective is to design an adaptive estimator to estimate the state $x$, online, using input-output measurements and to simultaneously estimate and utilize the optimal feedback controller that minimizes the cost functional 
\begin{equation}
J\left(x\left(\cdot\right),u\left(\cdot\right)\right)=\int_{0}^{\infty}r\left(x\left(\tau\right),u\left(\tau\right)\right)\D\tau,\label{eq:cost}
\end{equation}
while maintaining system stability during the learning phase. The function $ r:\R^{n\times m}\to\R $ is defined as $r\left(x,u\right)\coloneqq Q\left(x\right)+u^{T}Ru $, where $Q:\R^{n}\to\R$ is continuous, $R\in\R^{m\times m}$ is a constant positive definite matrix, and $ \gamma\geq0 $ is the discount factor.\begin{assumption}
	\label{ass:CostRestrictions}One of the following is true:\begin{enumerate}[label=(\alph*)]
	\item $ Q $ is positive definite.
	\item $ Q $ is positive semidefinite and $ p\mapsto Q\left(x\right) $ is positive definite for all nonzero $ q\in\R^{n} $.
	\item $ Q $ is positive semidefinite, $ q\mapsto Q\left(x\right) $ is positive definite for all nonzero $ p\in\R^{n} $ and $ f\left(x\right)\neq 0 $ whenever $ p\neq 0 $.
\end{enumerate}\end{assumption}To facilitate control design, the stability properties of the closed-loop system under optimal feedback are examined.
\section{Stability Under Optimal state Feedback\label{sec:Stability Under Optimal Feedback}}
The following theorem establishes global asymptotic stability of the closed-loop system under optimal state feedback.
\begin{thm}\label{thm:Optimal GAS}
	If the optimal state feedback controller $ u^{*}:\R^{2n}\to\R^{m} $ that minimizes the cost function in \eqref{eq:cost} exists and if the corresponding optimal value function $ V:\R^{2n}\to\R $ is continuously differentiable and radially unbounded, then the origin of closed-loop system \begin{align}
	\dot{y} & =q,\nonumber \\
	\dot{q} & =f\left(x\right)+g\left(x\right)u^{*}\left(x\right),\label{eq:Closed-loop System}
	\end{align} is globally asymptotically stable.
\end{thm}
\begin{proof}
	Under the hypothesis of Theorem \ref{thm:Optimal GAS}, the optimal value function is the unique solution of the Hamilton-Jacobi-Bellman equation \cite[pp. 164]{SCC.Liberzon2012}\begin{equation}
	V_{y}\left(x\right)q+V_{q}\left(x\right)\left(f\left(x\right)+g\left(x\right)u^{*}\left(x\right)\right)+r\left(y,u^{*}\left(x\right)\right)=0,\label{eq:HJB}
\end{equation}with\begin{equation}
u^{*}\left(x\right)=-\frac{1}{2}R^{-1}g^{T}\left(x\right)V_{q}\left(x\right),
\end{equation}
where the notation $ x_{y} $ denotes the partial derivative of $ x $ with respect to $ y $. The function $ V $ is positive semidefinite by definition. Since the solutions of \eqref{eq:Closed-loop System} are continuous, if $ V\left(\begin{bmatrix}
y\\q
\end{bmatrix}\right) = 0$ for some $ x \neq 0 $, it can be concluded that $ Q\left(\phi\left(t;x,u^{*}\left(\cdot\right)\right)\right) = 0, \forall t\geq 0$, and $ u^{*}\left(\phi\left(t;x,u^{*}\left(\cdot\right)\right)\right)= 0, \forall t\geq 0 $, where $ \phi\left(t,x,u\left(\cdot\right)\right) $ denotes the trajectory of \eqref{eq:Dynamics}, evaluated at time $ t $, starting from the state $ x $ and under the controller $ u\left(\cdot\right) $. If Assumption \ref{ass:CostRestrictions}-(a) holds then $ \phi\left(t;x,u^{*}\left(\cdot\right)\right)=0,\forall t\geq 0$, which contradicts $ x\neq 0 $. If Assumption \ref{ass:CostRestrictions}-(b) holds, then $ p\left(t;x,u^{*}\left(\cdot\right)\right)=0,\forall t\geq 0 $. As a result, $ \phi\left(t;x,u^{*}\left(\cdot\right)\right)=0,\forall t\geq 0$, which contradicts $ x\neq 0 $. If Assumption \ref{ass:CostRestrictions}-(c) holds, then $ q\left(t;x,u^{*}\left(\cdot\right)\right)=0,\forall t\geq 0 $. As a result, $ p\left(t;x,u^{*}\left(\cdot\right)\right)=c,\forall t\geq 0 $ for some constant $ c\in\R^{n} $. Since $ f\left(x\right)\neq 0 $ if $ p\neq 0 $, it can be concluded that $ c=0 $, which contradicts $ x\neq 0 $. Hence, $ V\left(x\right) $ cannot be zero for a nonzero $ x $. Furthermore, since $ f\left(0\right)=0 $, the zero controller is clearly the optimal controller starting from $ x=0 $. That is, $ V\left(0\right)=0 $, and as a result, $ V:\R^{2n}\to\R $ is positive definite.

Using $ V $ as a candidate Lyapunov function and using the HJB equation in \eqref{eq:HJB}, it can be concluded that\begin{equation*}
V_{y}\left(x\right)q+V_{q}\left(x\right)\left(f\left(x\right)+g\left(x\right)u^{*}\left(x\right)\right)\leq-Q\left(x\right),
\end{equation*}$ \forall x\in\R^{2n}. $ If Assumption \ref{ass:CostRestrictions}-(a) holds, then the proof is complete using Lyapunov's direct method. If Assumption \ref{ass:CostRestrictions}-(b) holds, then using the fact that if the output is identically zero then so is the state, the invariance principle \cite[Corollary 4.2]{SCC.Khalil2002} can be invoked to complete the proof. If Assumption \ref{ass:CostRestrictions}-(c) holds, then finiteness of the value function everywhere implies that the origin is the only equilibrium point of the closed-loop system. As a result, the invariance principle can be invoked to complete the proof.
\end{proof}
Using Theorem \ref{thm:Optimal GAS} and the converse Lyapunov theorem for asymptotic stability \cite[Theorem 4.17]{SCC.Khalil2002}, the existence of a radially unbounded positive definite function $ \mathcal{V}:\R^{2n}\to\R $ and a positive definite function $ W:\R^{2n}\to\R $ is guaranteed such that\begin{equation}\label{eq:Converse Lyapunov Function}
\mathcal{V}_{y}\left(x\right)q+\mathcal{V}_{q}\left(x\right)\left(f\left(x\right)+g\left(x\right)u^{*}\left(x\right)\right)\leq-W\left(x\right),
\end{equation}$ \forall x\in\R^{2n} $. The functions $ \mathcal{V} $ and $ W $ are utilized to analyze the stability of the output feedback approximate optimal controller. 
\section{Velocity Estimator Design\label{sec:Velocity-estimator-design}}
To generate estimates of the generalized velocity, a velocity estimator inspired by \cite{SCC.Dinh.Kamalapurkar.ea2014} is developed. The estimator is given by
\begin{align}
\dot{\hat{p}} & =\hat{q},\nonumber \\
\dot{\hat{q}} & =f\left(\hat{x}\right)+g\left(\hat{x}\right)u+\nu,\label{eq:Estimator}
\end{align}
where $\hat{x}$, $\hat{p}$, and $\hat{q}$ are estimates of $x$, $p$, and $q$, respectively, and $\nu$ is a feedback term designed in the following.

To facilitate the design of $\nu$, let $\tilde{p}=p-\hat{p}$, $\tilde{q}=q-\hat{q}$, and let
\begin{equation}
r=\dot{\tilde{p}}+\alpha\tilde{p}+\eta,\label{eq:r}
\end{equation}
where the signal $\eta$ is added to compensate for the fact that the generalized velocity state, $q$, is not measurable. Based on the subsequent stability analysis, the signal $\eta$ is designed as the output of the dynamic filter 
\begin{align}
\dot{\eta} & =-\beta\eta-kr-\alpha\tilde{q},\quad\eta\left(T_{0}\right)=0,\label{eq:eta Update}
\end{align}
where $\alpha,$ $k,$ and $\beta$ are positive constants and the feedback component $\nu$ is designed as 
\begin{equation}
\nu=\alpha^{2}\tilde{p}-\left(k+\alpha+\beta\right)\eta.\label{eq:Feedback}
\end{equation}
The design of the signals $\eta$ and $\nu$ to estimate the state from output measurements is inspired by the $p-$filter \cite{SCC.Xian.Queiroz.ea2004}. Using the fact that $\tilde{p}\left(0\right)=0$, the signal $\eta$ can be implemented via the integral form
\begin{equation}
\eta\left(t\right)=-\int\limits_{T_{0}}^{t}\left(\beta+k\right)\eta\left(\tau\right)\D\tau-\int\limits_{T_{0}}^{t}k\alpha\tilde{p}\left(\tau\right)\D\tau-\left(k+\alpha\right)\tilde{p}\left(t\right).\label{eq:IntegralUpdateEta}
\end{equation}
\section{Model-based Reinforcement Learning\label{sec:Model-based Reinforcement Learning}}
To estimate the optimal state feedback policy, the optimal value function, defined as\begin{equation*}
V\left(x\right)\coloneqq\min_{u\left(\cdot\right)}\int\limits_t^{\infty}r\left(\phi\left(\tau,x,u\left(\cdot\right)\right),u\left(\cdot\right)\right)\D\tau.
\end{equation*}The optimal value function $ V $ and the optimal policy $ u^{*} $ are approximated using parametric approximators $ \hat{V}:\R^{2n}\times\R^{L}\to\R $ and $ \hat{u}:\R^{2n}\times\R^{L}\to\R^{m} $ defined as\begin{align}
\hat{V}\left(x,W_{c}\right)&\coloneqq W_{c}^{T}\sigma\left(x\right), \mathrm{and}\\
\hat{u}\left(x,W_{a}\right)&\coloneqq-\frac{1}{2}R^{-1}g^{T}\left(x\right)\sigma_{x}^{T}\left(x\right)W_{a},
\end{align}where $ \sigma\coloneqq\left[\sigma_{1}\cdots,\sigma_{L}\right]$, $\sigma_{i}:\R^{2n}\to\R $ for all $ i $ is the vector of basis functions and $ W_{c}\in\R^{L} $ and $ W_{a}\in\R^{L} $ are estimates of the ideal parameters $ W\in\R^{L} $. The corresponding approximation error $ \epsilon:\R^{2n}\to\R $ is defined as $ \epsilon\left(x\right)\coloneqq V\left(x\right)-\hat{V}\left(x,W\right) $. Provided the basis functions are selected from an appropriate class of functions, for any given compact ball $ \overline{\B}\left(0,\chi\right)\subset\R^{2n} $, and any given $ \overline{\epsilon} $ there exists $ L\in\N $, a set of basis functions $ \left\{\sigma_{1},\cdots,\sigma_{L}\right\} $, and $ W\in\R^{L} $ such that $ \overline{\left\Vert \epsilon\right\Vert}_{\chi}<\overline{\epsilon} $ and $ \overline{\left\Vert \epsilon_{x}\right\Vert}_{\chi}<\overline{\epsilon} $, where $ \overline{\left\Vert \epsilon \right\Vert}_{\chi} $ denotes $ \sup_{x\in\overline{\B}\left(0,\chi\right)}\left\Vert \epsilon\left(x\right) \right\Vert $ (see \cite{SCC.Hornik.Stinchcombe.ea1990,SCC.Hornik1991,SCC.Abu-Khalaf.Lewis2005}).

Substituting the estimates $ \hat{V} $, $ \hat{u} $, and $ \hat{x} $ in \eqref{eq:HJB}, the Bellman error $ \delta:\R^{2n}\times\R^{L}\times\R^{L}\to\R $ is obtained as\begin{multline}
\delta\left(\hat{x},W_{c},W_{a}\right)=\hat{V}_{q}\left(\hat{x},W_{c}\right)\left(f\left(\hat{x}\right)+g\left(\hat{x}\right)\hat{u}\left(\hat{x},W_{a}\right)\right)\\+\hat{V}_{y}\left(\hat{x},W_{c}\right)\hat{q}+r\left(\hat{y},\hat{u}\left(\hat{x},W_{a}\right)\right),\label{eq:BE}
\end{multline}

Similar to \cite{SCC.Kamalapurkar.Walters.ea2016}, the technique developed in this result implements simulation of experience in a model-based RL scheme by using the system model to extrapolate the approximate BE to unexplored areas of the state space.  In the following, the trajectories of the state and the weight estimates $ W_{c} $ and $ W_{a} $, evaluated at time $ t $ starting from appropriate initial conditions are denoted by $ x\left(t\right) $, $ W_{c}\left(t\right) $ and $ W_{a}\left(t\right) $, respectively. The notation\footnote{For $a\in\R,$ the notation $\R_{\geq a}$ denotes the interval $\left[a,\infty\right)$ and the notation $\R_{>a}$ denotes the interval $\left(a,\infty\right)$.} $\delta_{t}:\R_{\geq 0}\to\R$ denotes the BE in \eqref{eq:BE}, evaluated along the trajectories of the state and the weight estimates as $\delta_{t}\left(t\right)\coloneqq\delta\left(\hat{x}\left(t\right),\hat{W}_{c}\left(t\right),\hat{W}_{a}\left(t\right)\right)$ and $\delta_{ti}:\R_{\geq 0}\to\R$ denotes BE extrapolated along the trajectories of the weight estimates and a predefined set of trajectories $\left\{ x_{i}:\R_{\geq 0}\to\R^{n}\mid i=1,\cdots,N\right\} $ as $\delta_{ti}\left(t\right)\coloneqq\delta\left(x_{i}\left(t\right),\hat{W}_{c}\left(t\right),\hat{W}_{a}\left(t\right)\right)$. A least-squares update law for the value function weights is designed based on the subsequent stability analysis as 
\begin{align}
\dot{\hat{W}}_{c} & =-\frac{k_{c}}{N}\Gamma\sum_{i=1}^{N}\frac{\omega_{i}}{\rho_{i}}\delta_{ti},\label{eq:OFBADP1criticupdate}\\
\dot{\Gamma} & =\beta\Gamma-\frac{k_{c}}{N}\Gamma\sum_{i=1}^{N}\frac{\omega_{i}\omega_{i}^{T}}{\rho_{i}^{2}}\Gamma,\label{eq:OFBADP1GammaDot}
\end{align}
$\Gamma\left(t_{0}\right)=\Gamma_{0},$ where $\Gamma:\R_{\geq t_{0}}\to\R^{L\times L}$ is a time-varying least-squares gain matrix, $\omega_{i}\left(t\right)\coloneqq\sigma_p\left(x_{i}\left(t\right)\right)q_{i}\left(t\right)+\sigma_q\left(x_{i}\left(t\right)\right)\left(f\left(x_{i}\left(t\right)\right)+g\left(x_{i}\left(t\right)\right)\hat{u}\left(x_{i}\left(t\right),W_{a}\left(t\right)\right)\right),$ $\rho_{i}\left(t\right)\coloneqq1+\gamma_{1}\omega_{i}^{T}\left(t\right)\omega_{i}\left(t\right)$, where $\gamma_{1}\in\R$ is a constant positive normalization gain, $\beta>0\in\R$ is a constant forgetting factor, and $k_{c}>0\in\R$ is a constant adaptation gain. 

The policy weights are updated to follow the value function weights as
\begin{multline}
\dot{W}_{a}=-k_{a1}\left(W_{a}-W_{c}\right)-k_{a2}W_{a}\\
+\sum_{i=1}^{N}\frac{k_{c}G_{i}^{T}W_{a}\omega_{i}^{T}}{4N\rho_{i}}W_{c},\label{eq:OFBADP1actorupdate}
\end{multline}
where $k_{a1},\:k_{a2}\in\R$ are positive constant adaptation gains,  $G_{i}\left(t\right)\coloneqq\sigma_{xi}\left(t\right)g_{i}\left(t\right)R^{-1}g_{i}^{T}\left(t\right)\sigma_{xi}^{T}\left(t\right)\in\R^{L\times L}$, $g_{i}\left(t\right)=g\left(x_{i}\left(t\right)\right)$ and $\sigma_{xi}\left(t\right)=\sigma_{x}\left(x_{i}\left(t\right)\right)$.
The following rank condition facilitates the subsequent analysis.
\begin{assumption}
	\label{ass:OFBADP1ADPLearnCond}There exists a finite set of trajectories $\left\{ x_{i}:\R_{\geq t_{0}}\to\R^{n}\mid i=1,\cdots,N\right\} $ and a constant $T\in\R_{>0}$ such that 
	\begin{gather}
	\underline{c}_{1}I_{L} \leq\inf_{t\in\R_{\geq t_{0}}}\left(\frac{1}{N}\sum_{i=1}^{N}\frac{\omega_{i}\left(t\right)\omega_{i}^{T}\left(t\right)}{\rho_{i}^{2}\left(t\right)}\right),\label{eq:OFBADP1PE1}\\
	\underline{c}_{2}I_{L}\leq \frac{1}{N}\int\limits_{t}^{t+T}\left(\sum_{i=1}^{N}\frac{\omega_{i}\left(\tau\right)\omega_{i}^{T}\left(\tau\right)}{\rho_{i}^{2}\left(\tau\right)}\right)\D\tau,\:\forall t\in\R_{\geq t_{0}},\label{eq:OFBADP1PE2}
	\end{gather}
	where, at least one of the nonnegative constants $\underline{c}_{1}$ and $\underline{c}_{2}$ is strictly positive.
\end{assumption}The rank conditions in (\ref{eq:OFBADP1PE1}) and (\ref{eq:OFBADP1PE2}) depend on the estimates $W_{a}$; hence, in general, they are impossible to guarantee a priori. However, unlike traditional adaptive dynamic programming literature that assumes that a regressor similar to $\omega_{i}$ evaluated along the system trajectories is PE, Assumption \ref{ass:OFBADP1ADPLearnCond} only requires the regressor $\omega_{i}$ to be persistently exciting. When the regressor is evaluated along the system state $x$ excitation in the regressor vanishes as the system states converge. Hence, in general, it is unlikely that a regressor evaluated along the system trajectories will be PE. However, the regressor $\omega_{i}$ depends on $x_{i}$, which can be designed independent of the system state $x$. Hence, $\underline{c}_{2}$ can be made strictly positive if the signal $x_{i}$ contains enough frequencies, and $\underline{c}_{1}$ can be made strictly positive by selecting a sufficient number of extrapolation trajectories, i.e., $N\gg L$. It is established in \cite[Lemma 1]{SCC.Kamalapurkar.Rosenfeld.ea2016} that under 
Assumption \ref{ass:OFBADP1ADPLearnCond} and provided $\lambda_{\min}\left\{ \Gamma_{0}^{-1}\right\} >0$, the update law in (\ref{eq:OFBADP1GammaDot}) ensures that the least squares gain matrix satisfies 
	\begin{equation}
	\underline{\Gamma}I_{L}\leq\Gamma\left(t\right)\leq\overline{\Gamma}I_{L},\label{eq:OFBADP1Gammabound}
	\end{equation}
	$\forall t\in\R_{\geq0}$ and for some $ \overline{\Gamma},\underline{\Gamma}>0 $.
\section{\label{sec:Analysis}Analysis}
The approximate BE, evaluated along the selected trajectories $\left\{ x_{i}\mid i=1,\cdots,N\right\} $, can be expressed as 
\begin{align}
\delta_{ti} & =-\omega_{i}^{T}\tilde{W}_{c}+\frac{1}{4}\tilde{W}_{a}^{T}G_{\sigma i}\tilde{W}_{a}+\Delta_{i},\label{eq:OFBADP1deltai2}
\end{align}
where $\nabla\epsilon_{i}=\nabla\epsilon\left(x_{i}\right)$, $f_{i}=f\left(x_{i}\right)$, $G_{i}\coloneqq g_{i}R^{-1}g_{i}^{T}\in\R^{n\times n}$, $\Delta_{i}\coloneqq\frac{1}{2}W^{T}\nabla\sigma_{i}G_{i}\nabla\epsilon_{i}^{T}+\frac{1}{4}G_{\epsilon i}-\nabla\epsilon_{i}f_{i}\in\R$ is a constant, $G_{\epsilon i}\coloneqq\nabla\epsilon_{i}G_{i}\nabla\epsilon_{i}^{T}\in\R$, and $G_{\sigma i}$ was introduced in (\ref{eq:OFBADP1actorupdate}). Using \eqref{eq:OFBADP1deltai2}, the time-derivative of the Lyapunov function introduced in \eqref{eq:Converse Lyapunov Function} along the trajectories of \eqref{eq:Dynamics} under the controller $ u\left(t\right)=\hat{u}\left(\hat{x\left(t\right)},W_{a}\left(t\right)\right) $ is given by\begin{equation*}
\dot{\mathcal{V}}\left(x,t\right)=\mathcal{V}_{y}\left(x\right)q+\mathcal{V}_{q}\left(x\right)\left(f\left(x\right)+g\left(\hat{x}\right)\hat{u}\left(\hat{x},W_{a}\right)\right)
\end{equation*}Adding and subtracting $ \mathcal{V}_{q}\left(x\right)\left(g\left(x\right)u^{*}\left(x\right)\right) $,\begin{multline*}
\dot{\mathcal{V}}\left(x,t\right)=\mathcal{V}_{y}\left(x\right)q+\mathcal{V}_{q}\left(x\right)\left(f\left(x\right)+g\left(x\right)u^{*}\left(x\right)\right)\\+\mathcal{V}_{q}\left(x\right)\left(g\left(\hat{x}\right)\hat{u}\left(\hat{x},W_{a}\right)-g\left(x\right)u^{*}\left(x\right)\right)
\end{multline*}Using \eqref{eq:Converse Lyapunov Function}, the fact that $ g $ is bounded, the basis functions $ \sigma $ are bounded, and that the value function approximation error $ \epsilon $ and its derivative with respect to $ x $ are bounded on compact sets, the time-derivative can be bounded as\begin{equation*}
\dot{\mathcal{V}}\left(x,t\right)\leq-W\left(x\right)+\iota_{1}\overline{\epsilon}+\iota_{2}\left\Vert \tilde{x}\right\Vert \left\Vert \tilde{W}_{a}\right\Vert +\iota_{3}\left\Vert \tilde{W}_{a}\right\Vert +\iota_{4}\left\Vert \tilde{x}\right\Vert,
\end{equation*}for all $ t\geq0 $ and for all $ x\in\overline{\B}\left(0,\chi\right) $ and $ \hat{x} \in\R^{2n}$, where $ \chi\subset\R^{2n} $ is a compact set, $ \iota_{1},\cdots,\iota_{4} $ are positive constants, and $ \tilde{x} \coloneqq x-\hat{x}$.

Let $ \Theta\left(\tilde{W}_{c},\tilde{W}_{a},t\right)\coloneqq \frac{1}{2}\tilde{W}_{c}^{T}\Gamma^{-1}\left(t\right)\tilde{W}_{c}+\frac{1}{2}\tilde{W}_{a}^{T}\tilde{W}_{a} $
 The time-derivative of $ \Theta $ along the trajectories of \eqref{eq:OFBADP1criticupdate}-\eqref{eq:OFBADP1actorupdate} is given by\begin{multline*}
\dot{\Theta}\left(\tilde{W}_{c},\tilde{W}_{a},t\right)=-\tilde{W}_{c}^{T}\Gamma^{-1}\left(-\frac{k_{c}}{N}\Gamma\sum_{i=1}^{N}\frac{\omega_{i}}{\rho_{i}}\delta_{ti}\right)\\
  -\frac{1}{2}\tilde{W}_{c}^{T}\left(\Gamma^{-1}\beta-\frac{k_{c}}{N}\sum_{i=1}^{N}\frac{\omega_{i}\omega_{i}^{T}}{\rho_{i}^{2}}\right)\tilde{W}_{c}\\
  -\tilde{W}_{a}^{T}\left(-k_{a1}\left(W_{a}-W_{c}\right)-k_{a2}W_{a}+\sum_{i=1}^{N}\frac{k_{c}G_{i}^{T}W_{a}\omega_{i}^{T}}{4N\rho_{i}}W_{c}\right)
 \end{multline*}
Using \eqref{eq:BE}, \begin{multline*}
\dot{\Theta}\left(\tilde{W}_{c},\tilde{W}_{a},t\right)\leq-k_{c}\underline{c}\left\Vert \tilde{W}_{c}\right\Vert ^{2}-\left(k_{a1}+k_{a2}\right)\left\Vert \tilde{W}_{a}\right\Vert ^{2}\\+k_{c}\iota_{8}\overline{\epsilon}\left\Vert \tilde{W}_{c}\right\Vert +k_{c}\iota_{5}\left\Vert \tilde{W}_{a}\right\Vert ^{2}+\left(k_{c}\iota_{6}+k_{a1}\right)\left\Vert \tilde{W}_{c}\right\Vert \left\Vert \tilde{W}_{a}\right\Vert \\+\left(k_{c}\iota_{7}+k_{a2}\overline{W}\right)\left\Vert \tilde{W}_{a}\right\Vert,
\end{multline*}for all $ t\geq0 $ and for all $ x\in\overline{\B}\left(0,\chi\right) $, where $ \iota_{5},\cdots,\iota_{8} $ are positive constants that are independent of the learning gains, $ \overline{W} $ denotes an upper bound on the norm of the ideal weights $ W $, and $ \underline{c} =\min_{t\geq 0} \lambda_{\min}\left\{\left(\frac{\beta}{2k_{c}}\Gamma^{-1}\left(t\right)+\frac{1}{2N}\sum_{i=1}^{N}\frac{\omega_{i}\omega_{i}^{T}}{\rho_{i}}\right)\right\}$. Assumption \ref{ass:OFBADP1ADPLearnCond} and \eqref{eq:OFBADP1Gammabound} guarantee that $ \underline{c}>0 $.

Let $ \Phi\left(\tilde{p},r,\eta\right)\coloneqq\frac{\alpha^{2}}{2}\tilde{p}^{T}\tilde{p}+\frac{1}{2}r^{T}r+\frac{1}{2}\eta^{T}\eta $. The time-derivative of $ \Phi $ along the trajectories of \eqref{eq:Dynamics} and \eqref{eq:Estimator}-\eqref{eq:Feedback} is given by\begin{multline*}
\dot{\Phi}\left(\tilde{p},r,\eta,t\right)=\alpha^{2}\tilde{p}^{T}\left(r-\alpha\tilde{p}-\eta\right)+\eta\left(-\beta\eta-kr-\alpha\tilde{q}\right)\\+r^{T}\left(\tilde{f}\left(x,\hat{x}\right)+\tilde{g}\left(x,\hat{x}\right)\hat{u}\left(\hat{x},W_{a}\right)-\alpha^{2}\tilde{p}-kr+k\eta+\alpha\eta\right),
\end{multline*}where $ \tilde{f}\left(x,\hat{x}\right)\coloneqq f\left(x\right)-f\left(\hat{x}\right) $ and $ \tilde{g}\left(x,\hat{x}\right)\coloneqq g\left(x\right)-g\left(\hat{x}\right)$. The time derivative of $ \Phi $ can be bounded above as\begin{multline*}
\dot{\Phi}\left(\tilde{p},r,\eta,t\right)\leq-\alpha^{3}\left\Vert \tilde{p}\right\Vert ^{2}-\left(k-\varpi_{1}\right)\left\Vert r\right\Vert ^{2}-\left(\beta-\alpha\right)\left\Vert \eta\right\Vert ^{2}\\
+\varpi_{1}\left(1+\alpha\right)\left\Vert r\right\Vert \left\Vert \tilde{p}\right\Vert +\varpi_{1}\left\Vert r\right\Vert \left\Vert \eta\right\Vert +\varpi_{3}\left\Vert r\right\Vert +\varpi_{2}\left\Vert r\right\Vert \left\Vert \tilde{W}_{a}\right\Vert 
\end{multline*}for all $ t\geq0 $ and for all $ x,\tilde{x}\in\overline{\B}\left(0,\chi\right) $, where $ \varpi_{1},\cdots,\varpi_{3} $ are positive constants that are independent of the learning gains.

The candidate Lyapunov function for the overall system is then defined as $ \mathscr{V}\left(Z,t\right)= \mathcal{V}\left(x\right)+\Theta\left(\tilde{W}_{c},\tilde{W}_{a},t\right)+\Phi\left(\tilde{p},r,\eta\right)$, where $ Z\coloneqq\begin{bmatrix}
x^{T}&\tilde{p}^{T}&r^{T}&\eta^{T}&\tilde{W}_{c}^{T}&\tilde{W}_{a}^{T}
\end{bmatrix}^{T} $. The time derivative of the candidate Lyapunov function can be bounded as \begin{equation*}
\dot{\mathscr{V}}\left(Z,t\right)\leq-W\left(x\right)-z^{T}\left(\frac{M+M^{T}}{2}\right)z+Pz+\iota_{1}\overline{\epsilon},
\end{equation*}where $ z\coloneqq \begin{bmatrix}\left\Vert \tilde{W}_{c}\right\Vert  & \left\Vert \tilde{W}_{a}\right\Vert  & \left\Vert \tilde{p}\right\Vert  & \left\Vert r\right\Vert  & \left\Vert \eta\right\Vert \end{bmatrix}^{T} $, $ P= $\begin{equation*}
\begin{bmatrix}k_{c}\iota_{8}\overline{\epsilon} & \left(k_{c}\iota_{7}+\iota_{3}+k_{a2}\overline{W}\right) & \iota_{4}\left(1+\alpha\right) & \left(\varpi_{3}+\iota_{4}\right) & \iota_{4}\end{bmatrix}
\end{equation*} $ M= $\footnotesize\begin{equation*}
\begin{bmatrix}k_{c}\underline{c} & -\left(k_{c}\iota_{6}+k_{a1}\right) & 0 & 0 & 0\\
0 & \left(k_{a1}+k_{a2}-k_{c}\iota_{5}\right) & -\iota_{2}\left(1+\alpha\right) & -\left(\iota_{2}+\varpi_{2}\right) & -\iota_{2}\\
0 & 0 & \alpha^{3} & -\varpi_{1}\left(1+\alpha\right) & 0\\
0 & 0 & 0 & \left(k-\varpi_{1}\right) & -\varpi_{1}\\
0 & 0 & 0 & 0 & \left(\beta-\alpha\right)
\end{bmatrix}.
\end{equation*}\normalsize Provided the matrix $ M+M^{T} $ is positive definite,\begin{equation*}
\dot{\mathscr{V}}\left(Z,t\right)\leq-W\left(x\right)-\underline{M}\left\Vert z\right\Vert ^{2}+\overline{P}\left\Vert z\right\Vert +\iota_{1}\overline{\epsilon},
\end{equation*}where $ \underline{M}\coloneqq\lambda_{\min}\left\{\frac{M+M^{T}}{2}\right\} $. Letting $ \underline{M}\eqqcolon\underline{M}_{1}+\underline{M}_{2} $ and letting $ \mathcal{W}:\R^{5*n+2*L}\to\R $ be defined as $ \mathcal{W}\left(Z\right)=-W\left(x\right)-\underline{M}_{1}\left\Vert z\right\Vert^{2}$, the bound\begin{equation}
\dot{\mathscr{V}}\left(Z,t\right)\leq-\mathcal{W}\left(Z\right),\forall\left\Vert Z\right\Vert \geq \mu, Z\in\overline{\B}\left(0,\frac{\chi}{3\left(1+\alpha\right)}\right),\label{eq:OFBADPVDotBound}
\end{equation}for all $ t\geq0 $. 

Using the bound in \eqref{eq:OFBADP1Gammabound} and the fact that the converse Lyapunov function is guaranteed to be time-independent, radially unbounded, and positive definite, Lemma 4.3 can be invoked to conclude that \begin{equation}
\underline{v}\left(\left\Vert Z\right\Vert \right)\leq V_{L}\left(Z,t\right)\leq\overline{v}\left(\left\Vert Z\right\Vert \right),\label{eq:OFBADPVBound}
\end{equation}
for all $t\in\R_{\geq 0}$ and for all $Z\in\R^{5n+2L}$, where $\underline{v},\overline{v}:\R_{\geq0}\rightarrow\R_{\geq0}$ are class $\mathcal{K}$ functions.

Provided the learning gains, the domain radius $ \chi $, and the basis functions for function approximation are selected such that $ M+M^{T} $ is positive definite and $ \mu<\overline{v}^{-1}\left(\underline{v}\left(\frac{\chi}{4\left(1+\alpha\right)}\right)\right) $, Theorem 4.18 in \cite{SCC.Khalil2002} can be invoked to conclude that Z is uniformly ultimately bounded. Since the estimates $W_{a}$ approximate the ideal weights $ W $, the policy $ \hat{u} $ approximates the optimal policy $ u^{*} $.

\begin{figure}
	\includegraphics[width=1\columnwidth]{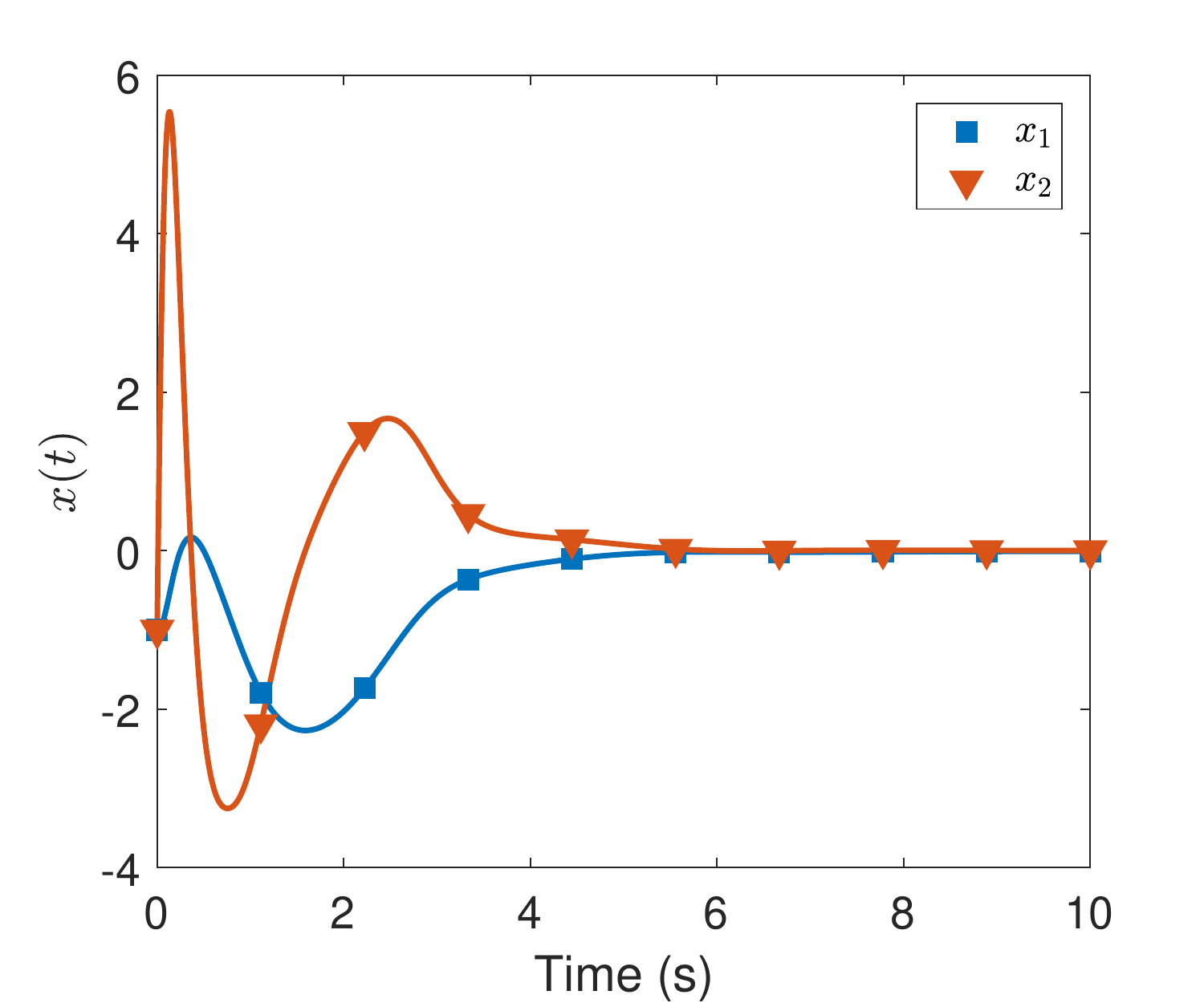}
	\caption{System state trajectories generated using the developed technique.}
	\label{fig:state}
\end{figure}
\begin{figure}
	\includegraphics[width=1\columnwidth]{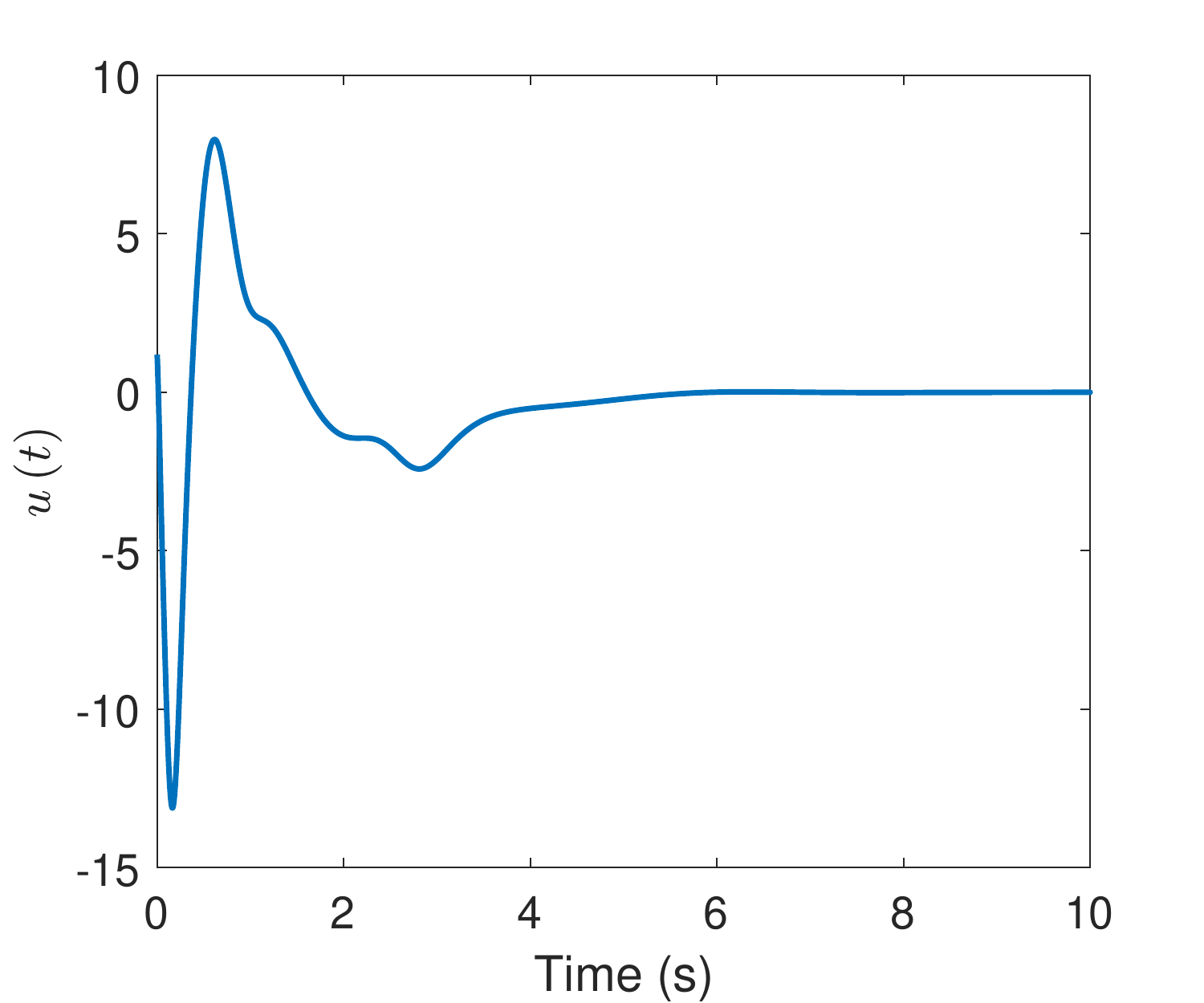}
	\caption{Control trajectories generated using the developed technique.}
	\label{fig:control}
\end{figure}
\begin{figure}
	\includegraphics[width=1\columnwidth]{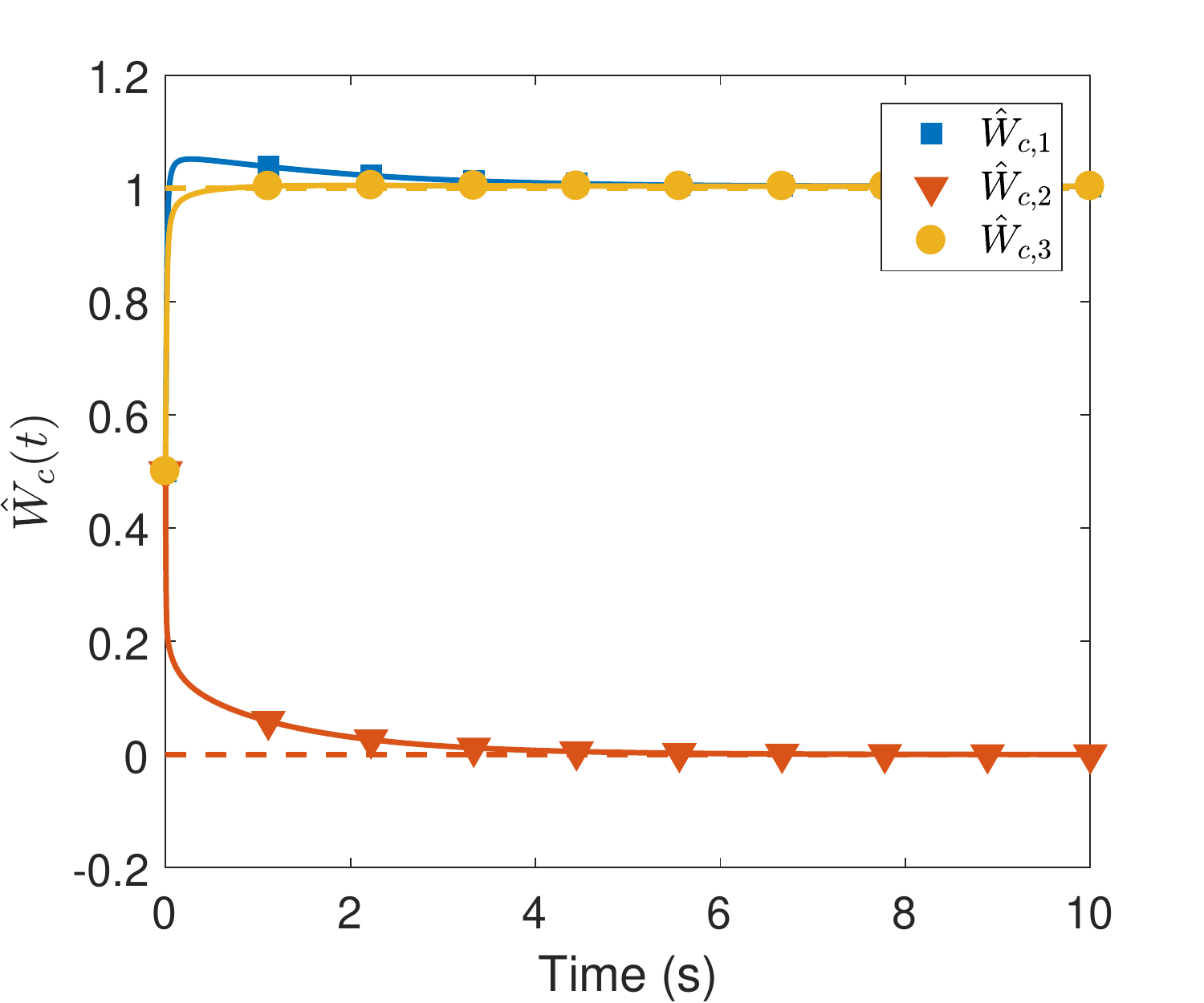}
	\caption{Critic weight estimates generated using the developed technique, and compared to the ideal values (marked with dashed lines).}
	\label{fig:wch}
\end{figure}
\begin{figure}
	\includegraphics[width=1\columnwidth]{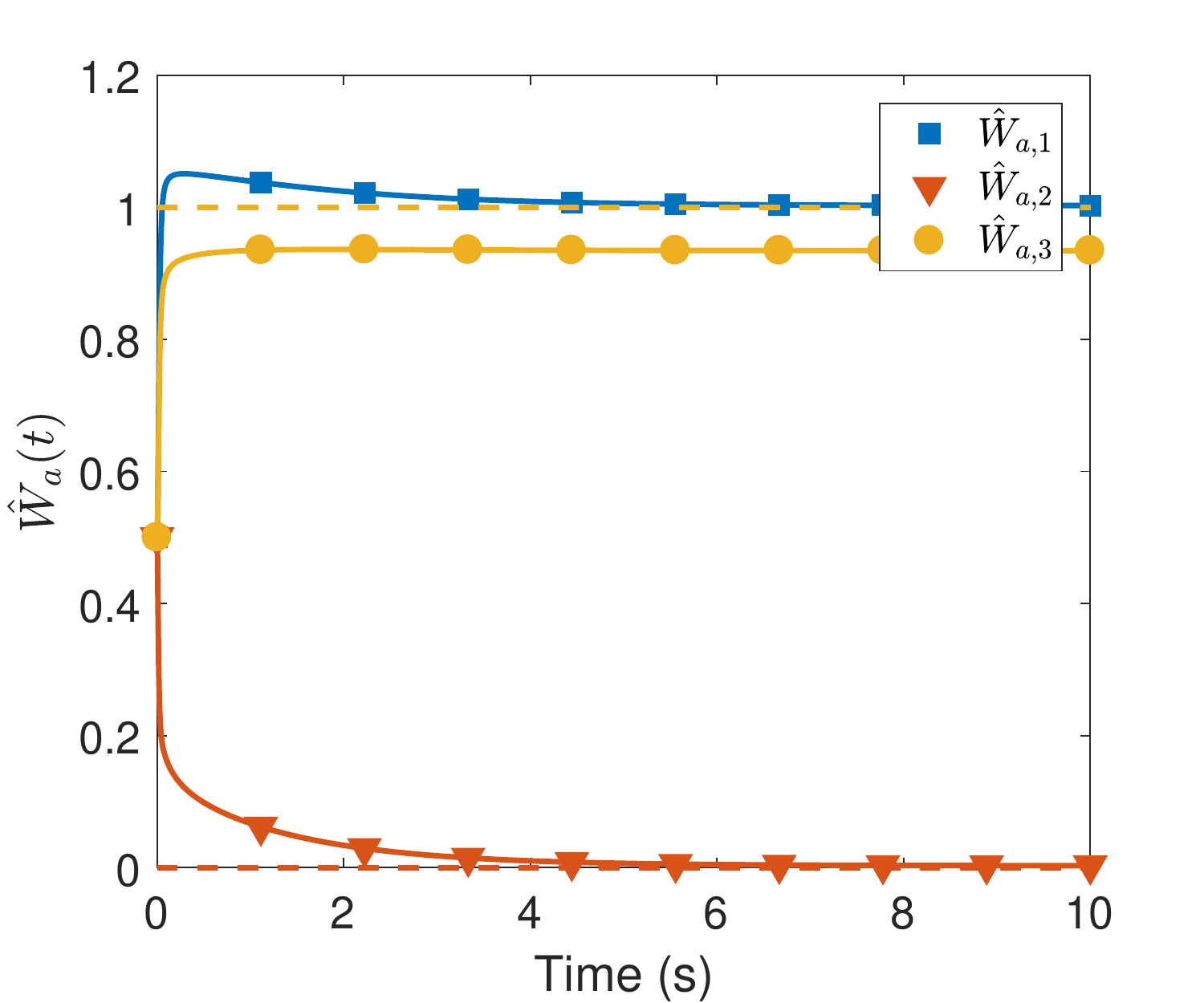}
	\caption{Actor weight estimates generated using the developed technique, and compared to the ideal values (marked with dashed lines).}
	\label{fig:wah}
\end{figure}
\begin{figure}
	\includegraphics[width=1\columnwidth]{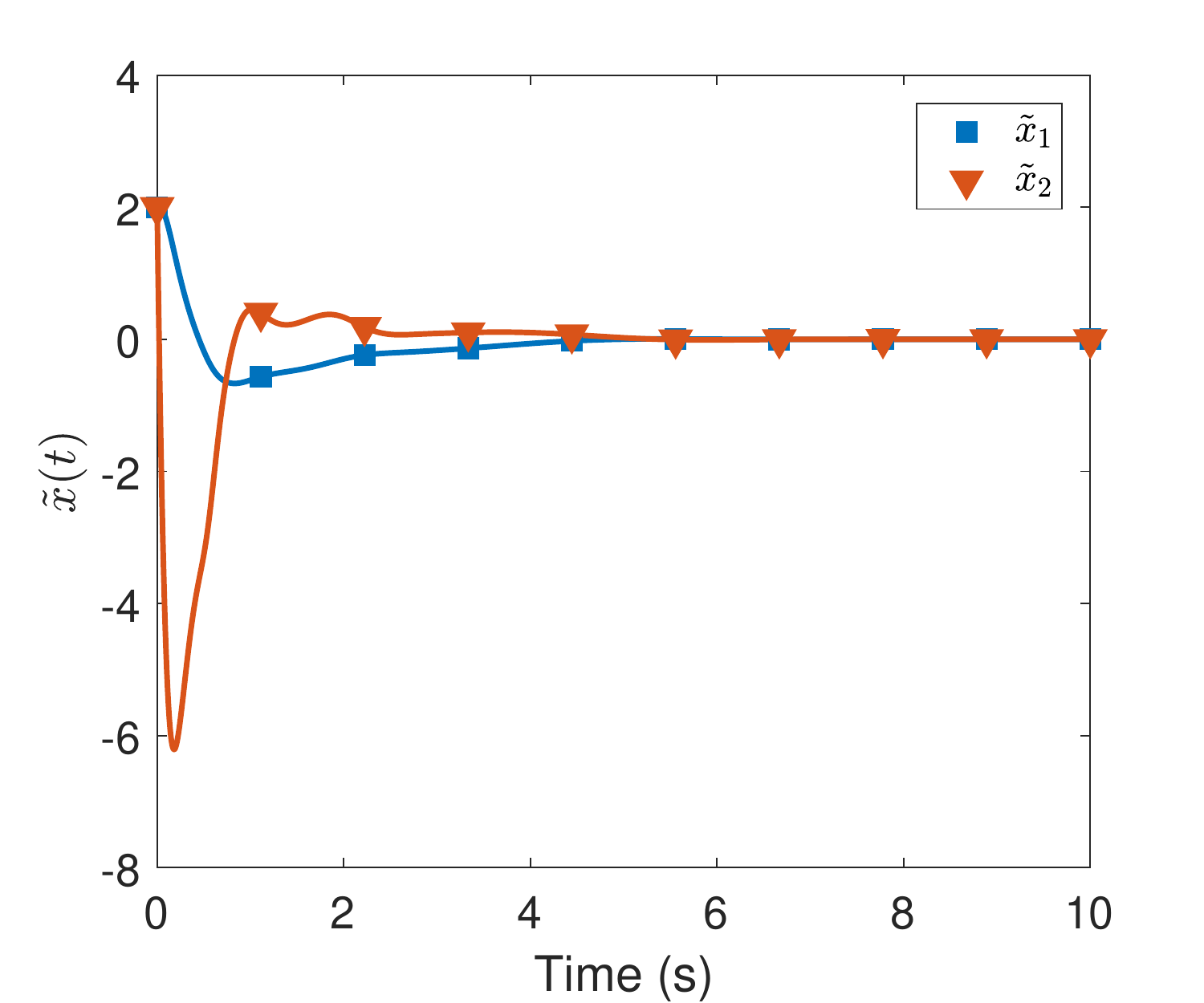}
	\caption{State estimation error.}
	\label{fig:xtilde}
\end{figure}
\section{\label{sec:Simulation Results}Simulation Results}
The performance of the developed controller is demonstrated by simulating a nonlinear, control affine system with a two dimensional state $ x=[x_{1},\:x_{2}]^{T} $. The system dynamics are described by \eqref{eq:Dynamics} where \begin{align*}
f\left(x\right)&=-x_{1}-\frac{1}{2}x_{2}\left(1-\left(\cos\left(2x_{1}\right)+2\right)^{2}\right)\\g\left(x\right)&=\cos\left(2x_{1}\right)+2.
\end{align*}The origin is an unstable equilibrium point of the unforced system $\dot{x}=f\left(x\right)$. The control objective is to minimize the cost in \eqref{eq:cost}, where $ Q\left(x\right)=q^{2} $ and $ R=1 $. For comparison purposes, the optimal value function for this problem is computed using the converse method in \cite{SCC.Nevistic.Primbs1996} as $ V^{*}\left(x\right)=x_{1}^{2}+x_{2}^{2} $.

The basis function $\sigma:\R^{2}\to\R^{3}$ for value function approximation is selected as $\sigma=\left[x_{1}^{2}, x_{1}x_{2}, x_{2}^{2}\right]^{T}$. Based on the analytical solution, the ideal weights are $W=\left[1,\:0,\:1\right]^{T}$. The data points for the simulation of experience in the update law \eqref{eq:OFBADP1criticupdate} are selected to be on a $5\times5$ grid around the origin. The learning gains are selected as $k_{c}=0.2$, $k_{a1}=100$, $k_{a2}=0.1$, $ \beta_{\gamma}=3 $, and $\nu=0.005$. The gains for the state estimator are selected as $ k=5 $, $ \alpha=0.2 $, and $ \beta=5 $. The initial conditions are selected as $ x\left(0\right)=[1,1]^{T} $, $ \hat{x}\left(0\right)=[-1,-1]^{T} $, $ W_{a}\left(0\right)=W_{c}\left(0\right)=[0.5,0.5,0.5]^{T} $, and $ \Gamma\left(0\right)=50\id_{3} $.

Figs. \ref{fig:state}-\ref{fig:xtilde} demonstrates that the system state is regulated to the origin, the generalized velocities are identified, and the actor and the critic weights converge to their true values. Furthermore, unlike previous results, a probing signal to ensure persistence of excitation\index{persistence of excitation} is not required.
\section{\label{sec:Conclusion}Conclusion}
An output-feedback MBRL method is developed for a class of second-order nonlinear systems. The control technique uses exact model knowledge and integrates a dynamic state estimator within the model-based reinforcement learning framework to achieve output-feedback MBRL. Simulation results demonstrate the efficacy of the developed method. 
Integration of simultaneous state and parameter estimation methods such as \cite{SCC.Kamalapurkar2017a} with the MBRL method to achieve output-feedback MBRL using uncertain models is a topic for future research.
\bibliographystyle{IEEEtran}
\bibliography{encr,sccmaster,scc}

\begin{thebibliography}{10}
\def\url#1{}
\csname url@samestyle\endcsname
\providecommand{\newblock}{\relax}
\providecommand{\bibinfo}[2]{#2}
\providecommand{\BIBentrySTDinterwordspacing}{\spaceskip=0pt\relax}
\providecommand{\BIBentryALTinterwordstretchfactor}{4}
\providecommand{\BIBentryALTinterwordspacing}{\spaceskip=\fontdimen2\font plus
\BIBentryALTinterwordstretchfactor\fontdimen3\font minus
  \fontdimen4\font\relax}
\providecommand{\BIBforeignlanguage}[2]{{%
\expandafter\ifx\csname l@#1\endcsname\relax
\typeout{** WARNING: IEEEtran.bst: No hyphenation pattern has been}%
\typeout{** loaded for the language `#1'. Using the pattern for}%
\typeout{** the default language instead.}%
\else
\language=\csname l@#1\endcsname
\fi
#2}}
\providecommand{\BIBdecl}{\relax}
\BIBdecl

\bibitem{SCC.Chen.Jagannathan2008}
Z.~Chen and S.~Jagannathan, ``Generalized {H}amilton-{J}acobi-{B}ellman
  formulation -based neural network control of affine nonlinear discrete-time
  systems,'' \emph{IEEE Trans. Neural Netw.}, vol.~19, no.~1, pp. 90--106, Jan.
  2008.

\bibitem{SCC.Mehta.Meyn2009}
P.~Mehta and S.~Meyn, ``{Q}-learning and pontryagin's minimum principle,'' in
  \emph{Proc. IEEE Conf. Decis. Control}, Dec. 2009, pp. 3598--3605.

\bibitem{SCC.Vrabie.Lewis2010}
D.~Vrabie and F.~L. Lewis, ``Integral reinforcement learning for online
  computation of feedback nash strategies of nonzero-sum differential games,''
  in \emph{Proc. IEEE Conf. Decis. Control}, 2010, pp. 3066--3071.

\bibitem{SCC.Vamvoudakis.Lewis2010}
K.~G. Vamvoudakis and F.~L. Lewis, ``{Online actor-critic algorithm to solve
  the continuous-time infinite horizon optimal control problem},''
  \emph{Automatica}, vol.~46, no.~5, pp. 878--888, 2010.

\bibitem{SCC.Lewis.Vrabie.ea2012}
F.~L. Lewis, D.~Vrabie, and V.~L. Syrmos, \emph{Optimal control}, 3rd~ed.\hskip
  1em plus 0.5em minus 0.4em\relax Hoboken, NJ: Wiley, 2012.

\bibitem{SCC.Lee.Park.ea2012}
J.~Y. Lee, J.~B. Park, and Y.~H. Choi, ``Integral {Q}-learning and explorized
  policy iteration for adaptive optimal control of continuous-time linear
  systems,'' \emph{Automatica}, vol.~48, no.~11, pp. 2850--2859, Nov. 2012.

\bibitem{SCC.Modares.Lewis.ea2013}
H.~Modares, F.~L. Lewis, and M.-B. Naghibi-Sistani, ``Adaptive optimal control
  of unknown constrained-input systems using policy iteration and neural
  networks,'' \emph{IEEE Trans. Neural Netw. Learn. Syst.}, vol.~24, no.~10,
  pp. 1513--1525, 2013.

\bibitem{SCC.Bhasin.Kamalapurkar.ea2013a}
\BIBentryALTinterwordspacing
S.~Bhasin, R.~Kamalapurkar, M.~Johnson, K.~G. Vamvoudakis, F.~L. Lewis, and
  W.~E. Dixon, ``A novel actor-critic-identifier architecture for approximate
  optimal control of uncertain nonlinear systems,'' \emph{Automatica}, vol.~49,
  no.~1, pp. 89--92, Jan. 2013.
  \url{http://www.sciencedirect.com/science/article/pii/S0005109812004827}
\BIBentrySTDinterwordspacing

\bibitem{SCC.Bian.Jiang.ea2014}
T.~Bian, Y.~Jiang, and Z.-P. Jiang, ``Adaptive dynamic programming and optimal
  control of nonlinear nonaffine systems,'' \emph{Automatica}, vol.~50, no.~10,
  pp. 2624--2632, 2014.

\bibitem{SCC.Kiumarsi.Lewis.ea2014}
B.~Kiumarsi, F.~L. Lewis, H.~Modares, A.~Karimpour, and M.-B. Naghibi-Sistani,
  ``Reinforcement {Q}-learning for optimal tracking control of linear
  discrete-time systems with unknown dynamics,'' \emph{Automatica}, vol.~50,
  no.~4, pp. 1167--1175, Apr. 2014.

\bibitem{SCC.Modares.Lewis2014}
H.~Modares and F.~L. Lewis, ``Optimal tracking control of nonlinear
  partially-unknown constrained-input systems using integral reinforcement
  learning,'' \emph{Automatica}, vol.~50, no.~7, pp. 1780--1792, 2014.

\bibitem{SCC.Bian.Jiang.ea2015}
T.~Bian, Y.~Jiang, and Z.-P. Jiang, ``Decentralized adaptive optimal control of
  large-scale systems with application to power systems,'' \emph{IEEE Trans.
  Ind. Electron.}, vol.~62, no.~4, pp. 2439--2447, Apr. 2015.

\bibitem{SCC.Li.Liu.ea2015}
C.~Li, D.~Liu, and H.~Li, ``Finite horizon optimal tracking control of
  partially unknown linear continuous-time systems using policy iteration,''
  \emph{IET Control Theory Appl.}, vol.~9, no.~12, pp. 1791--1801, 2015.

\bibitem{SCC.Yang.Liu.ea2015}
X.~Yang, D.~Liu, Q.~Wei, and D.~Wang, ``Direct adaptive control for a class of
  discrete-time unknown nonaffine nonlinear systems using neural networks,''
  \emph{Int. J. Robust Nonlinear Control}, vol.~25, no.~12, pp. 1844--1861,
  Apr. 2015.

\bibitem{SCC.Zhao.Xu.ea2015}
Q.~Zhao, H.~Xu, and S.~Jagannathan, ``Neural network-based finite-horizon
  optimal control of uncertain affine nonlinear discrete-time systems,''
  \emph{IEEE Trans. Neural Netw. Learn. Syst.}, vol.~26, no.~3, pp. 486--499,
  2015.

\bibitem{SCC.Cichosz1999}
P.~Cichosz, ``An analysis of experience replay in temporal difference
  learning,'' \emph{Cybern. Syst.}, vol.~30, no.~5, pp. 341--363, 1999.

\bibitem{SCC.Wawrzynski2009}
P.~Wawrzy{\'n}ski, ``Real-time reinforcement learning by sequential
  actor-critics and experience replay,'' \emph{Neural Netw.}, vol.~22, no.~10,
  pp. 1484--1497, 2009.

\bibitem{SCC.Zhang.Cui.ea2011}
H.~Zhang, L.~Cui, X.~Zhang, and Y.~Luo, ``Data-driven robust approximate
  optimal tracking control for unknown general nonlinear systems using adaptive
  dynamic programming method,'' \emph{IEEE Trans. Neural Netw.}, vol.~22,
  no.~12, pp. 2226--2236, Dec. 2011.

\bibitem{SCC.Adam.Busoniu.ea2012}
S.~Adam, L.~Busoniu, and R.~Babuska, ``Experience replay for real-time
  reinforcement learning control,'' \emph{IEEE Trans. Syst. Man Cybern. Part C
  Appl. Rev.}, vol.~42, no.~2, pp. 201--212, 2012.

\bibitem{SCC.Luo.Wu.ea2014}
B.~Luo, H.-N. Wu, T.~Huang, and D.~Liu, ``Data-based approximate policy
  iteration for affine nonlinear continuous-time optimal control design,''
  \emph{Automatica}, 2014.

\bibitem{SCC.Modares.Lewis.ea2014}
H.~Modares, F.~L. Lewis, and M.-B. Naghibi-Sistani, ``Integral reinforcement
  learning and experience replay for adaptive optimal control of
  partially-unknown constrained-input continuous-time systems,''
  \emph{Automatica}, vol.~50, no.~1, pp. 193--202, 2014.

\bibitem{SCC.Sutton1991a}
\BIBentryALTinterwordspacing
R.~S. Sutton, ``Integrated modeling and control based on reinforcement learning
  and dynamic programming,'' in \emph{Advances in Neural Information Processing
  Systems 3}, R.~P. Lippmann, J.~E. Moody, and D.~S. Touretzky, Eds.\hskip 1em
  plus 0.5em minus 0.4em\relax Morgan-Kaufmann, 1991, pp. 471--478.
  \url{http://papers.nips.cc/paper/388-integrated-modeling-and-control-based-on-reinforcement-learning-and-dynamic-programming.pdf}
\BIBentrySTDinterwordspacing

\bibitem{SCC.Lampe.Riedmiller2014}
T.~Lampe and M.~Riedmiller, ``Approximate model-assisted neural fitted
  {Q}-iteration,'' in \emph{Int. Joint Conf. Neural Netw.}, 2014, pp.
  2698--2704.

\bibitem{SCC.Abbeel.Quigley.ea2006}
P.~Abbeel, M.~Quigley, and A.~Y. Ng, ``Using inaccurate models in reinforcement
  learning,'' in \emph{Proc. Int. Conf. Mach. Learn.}\hskip 1em plus 0.5em
  minus 0.4em\relax New York, NY, USA: ACM, 2006, pp. 1--8.

\bibitem{SCC.Deisenroth.Rasmussen2011}
M.~P. Deisenroth and C.~E. Rasmussen, ``Pilco: a model-based and data-efficient
  approach to policy search,'' in \emph{Proc. Int. Conf. Mach. Learn.}, 2011,
  pp. 465--472.

\bibitem{SCC.Heess.Wayne.ea2015}
\BIBentryALTinterwordspacing
N.~Heess, G.~Wayne, D.~Silver, T.~Lillicrap, T.~Erez, and Y.~Tassa, ``Learning
  continuous control policies by stochastic value gradients,'' in
  \emph{Advances in Neural Information Processing Systems 28}, C.~Cortes, N.~D.
  Lawrence, D.~D. Lee, M.~Sugiyama, and R.~Garnett, Eds.\hskip 1em plus 0.5em
  minus 0.4em\relax Curran Associates, Inc., 2015, pp. 2944--2952.
  \url{http://papers.nips.cc/paper/5796-learning-continuous-control-policies-by-stochastic-value-gradients.pdf}
\BIBentrySTDinterwordspacing

\bibitem{SCC.Grondman2015}
I.~Grondman, ``Online model learning algorithms for actor-critic control,''
  Ph.D. dissertation, Technische Universiteit Delft, 2015.

\bibitem{SCC.Brafman.Tennenholtz2002}
R.~I. Brafman and M.~Tennenholtz, ``{R-MAX} - a general polynomial time
  algorithm for near-optimal reinforcement learning,'' \emph{J. Mach. Learn.
  Res.}, vol.~3, pp. 213--231, Oct. 2002.

\bibitem{SCC.Kearns.Singh2002}
\BIBentryALTinterwordspacing
M.~Kearns and S.~Singh, ``Near-optimal reinforcement learning in polynomial
  time,'' \emph{Machine Learning}, vol.~49, no.~2, pp. 209--232, Nov. 2002.
  \url{https://doi.org/10.1023/A:1017984413808}
\BIBentrySTDinterwordspacing

\bibitem{SCC.Kakade.Kearns.ea2003}
\BIBentryALTinterwordspacing
S.~Kakade, M.~J. Kearns, and J.~Langford, ``Exploration in metric state
  spaces,'' in \emph{Proc. Int. Conf. Mach. Learn.}, T.~Fawcett and N.~Mishra,
  Eds., 2003, pp. 306--312.
  \url{http://www.aaai.org/Papers/ICML/2003/ICML03-042.pdf}
\BIBentrySTDinterwordspacing

\bibitem{SCC.Nouri.Littman2009}
\BIBentryALTinterwordspacing
A.~Nouri and M.~L. Littman, ``Multi-resolution exploration in continuous
  spaces,'' in \emph{Advances in Neural Information Processing Systems 21},
  D.~Koller, D.~Schuurmans, Y.~Bengio, and L.~Bottou, Eds.\hskip 1em plus 0.5em
  minus 0.4em\relax Curran Associates, Inc., 2009, pp. 1209--1216.
  \url{http://papers.nips.cc/paper/3557-multi-resolution-exploration-in-continuous-spaces.pdf}
\BIBentrySTDinterwordspacing

\bibitem{SCC.Li.Littman.ea2011}
L.~Li, M.~L. Littman, T.~J. Walsh, and A.~L. Strehl, ``Knows what it knows: a
  framework for self-aware learning,'' \emph{Mach. Learn.}, vol.~82, no.~3, pp.
  399--443, 2011.

\bibitem{SCC.Jung.Stone2010}
\BIBentryALTinterwordspacing
T.~Jung and P.~Stone, ``{G}aussian processes for sample efficient reinforcement
  learning with {RMAX}-like exploration,'' in \emph{Machine Learning and
  Knowledge Discovery in Databases, ECML PKDD 2010}, ser. Lecture Notes in
  Computer Science, J.~L. Balc{\'a}zar, F.~Bonchi, A.~Gionis, and M.~Sebag,
  Eds.\hskip 1em plus 0.5em minus 0.4em\relax Berlin, Heidelberg: Springer
  Berlin Heidelberg, 2010, vol. 6321, pp. 601--616.
  \url{https://doi.org/10.1007/978-3-642-15880-3_44}
\BIBentrySTDinterwordspacing

\bibitem{SCC.Grande.Walsh.ea2014}
\BIBentryALTinterwordspacing
R.~Grande, T.~Walsh, and J.~How, ``Sample efficient reinforcement learning with
  {G}aussian processes,'' in \emph{Proc. Int. Conf. Mach. Learn.}, E.~P. Xing
  and T.~Jebara, Eds., vol.~32, no.~2, Jun. 2014, pp. 1332--1340.
  \url{http://proceedings.mlr.press/v32/grande14.html}
\BIBentrySTDinterwordspacing

\bibitem{SCC.Kamalapurkar.Dinh.ea2015}
\BIBentryALTinterwordspacing
R.~Kamalapurkar, H.~Dinh, S.~Bhasin, and W.~E. Dixon, ``Approximate optimal
  trajectory tracking for continuous-time nonlinear systems,''
  \emph{Automatica}, vol.~51, pp. 40--48, Jan. 2015.
  \url{http://www.sciencedirect.com/science/article/pii/S0005109814004841}
\BIBentrySTDinterwordspacing

\bibitem{SCC.Kamalapurkar.Walters.ea2016}
\BIBentryALTinterwordspacing
R.~Kamalapurkar, P.~Walters, and W.~E. Dixon, ``Model-based reinforcement
  learning for approximate optimal regulation,'' \emph{Automatica}, vol.~64,
  pp. 94--104, Feb. 2016.
  \url{http://www.sciencedirect.com/science/article/pii/S0005109815004392}
\BIBentrySTDinterwordspacing

\bibitem{SCC.Kamalapurkar.Walters.ea2016a}
\BIBentryALTinterwordspacing
------, ``Model-based reinforcement learning for approximate optimal
  regulation,'' in \emph{Control of Complex Systems: Theory and Applications},
  K.~Vamvoudakis and S.~Jagannathan, Eds.\hskip 1em plus 0.5em minus
  0.4em\relax Butterworth-Heinemann, Aug. 2016, pp. 247--273.
  \url{http://www.sciencedirect.com/science/article/pii/B9780128052464000082}
\BIBentrySTDinterwordspacing

\bibitem{SCC.Kamalapurkar.Rosenfeld.ea2016}
\BIBentryALTinterwordspacing
R.~Kamalapurkar, J.~A. Rosenfeld, and W.~E. Dixon, ``Efficient model-based
  reinforcement learning for approximate online optimal control,''
  \emph{Automatica}, vol.~74, pp. 247--258, Dec. 2016.
  \url{http://www.sciencedirect.com/science/article/pii/S0005109816303272}
\BIBentrySTDinterwordspacing

\bibitem{SCC.Kamalapurkar.Andrews.ea2017}
\BIBentryALTinterwordspacing
R.~Kamalapurkar, L.~Andrews, P.~Walters, and W.~E. Dixon, ``Model-based
  reinforcement learning for infinite-horizon approximate optimal tracking,''
  \emph{IEEE Trans. Neural Netw. Learn. Syst.}, vol.~28, no.~3, pp. 753--758,
  Mar. 2017.  \url{http://ieeexplore.ieee.org/document/7398079/}
\BIBentrySTDinterwordspacing

\bibitem{SCC.Kamalapurkar.Klotz.ea2018}
\BIBentryALTinterwordspacing
R.~Kamalapurkar, J.~R. Klotz, P.~Walters, and W.~E. Dixon, ``Model-based
  reinforcement learning in differential graphical games,'' \emph{IEEE Trans.
  Control Netw. Syst.}, vol.~5, no.~1, pp. 423--433, Mar. 2018.
  \url{http://ieeexplore.ieee.org/document/7590053/}
\BIBentrySTDinterwordspacing

\bibitem{SCC.Lewis.Vrabie2009}
F.~L. Lewis and D.~Vrabie, ``Reinforcement learning and adaptive dynamic
  programming for feedback control,'' \emph{IEEE Circuits Syst. Mag.}, vol.~9,
  no.~3, pp. 32--50, 2009.

\bibitem{SCC.Vrabie.Lewis2009}
D.~Vrabie and F.~L. Lewis, ``Neural network approach to continuous-time direct
  adaptive optimal control for partially unknown nonlinear systems,''
  \emph{Neural Netw.}, vol.~22, no.~3, pp. 237--246, 2009.

\bibitem{SCC.Vamvoudakis.Lewis2011a}
K.~G. Vamvoudakis and F.~L. Lewis, ``Multi-player non-zero-sum games: online
  adaptive learning solution of coupled {H}amilton-{J}acobi equations,''
  \emph{Automatica}, vol.~47, pp. 1556--1569, 2011.

\bibitem{SCC.Vamvoudakis.Lewis.ea2012}
\BIBentryALTinterwordspacing
K.~G. Vamvoudakis, F.~L. Lewis, and G.~R. Hudas, ``Multi-agent differential
  graphical games: online adaptive learning solution for synchronization with
  optimality,'' \emph{Automatica}, vol.~48, no.~8, pp. 1598--1611, 2012.
  \url{http://www.sciencedirect.com/science/article/pii/S0005109812002476}
\BIBentrySTDinterwordspacing

\bibitem{SCC.Song.Lewis.ea2015}
R.~Song, F.~L. Lewis, Q.~Wei, H.-G. Zhang, Z.-P. Jiang, and D.~Levine,
  ``Multiple actor-critic structures for continuous-time optimal control using
  input-output data,'' \emph{IEEE Trans. Neural Netw. Learn. Syst.}, vol.~26,
  no.~4, pp. 851--865, Apr. 2015.

\bibitem{SCC.Papadimitriou.Tsitsiklis1987}
C.~H. Papadimitriou and J.~N. Tsitsiklis, ``The complexity of {M}arkov decision
  processes,'' \emph{Math. Oper. Res.}, vol.~12, no.~3, pp. 441--450, 1987.

\bibitem{SCC.Madani.Hanks.ea2003}
O.~Madani, S.~Hanks, and A.~Condon, ``On the undecidability of probabilistic
  planning and related stochastic optimization problems,'' \emph{Artif.
  Intell.}, vol. 147, no. 1-2, pp. 5--34, 2003.

\bibitem{SCC.Modares.Lewis.ea2016}
\BIBentryALTinterwordspacing
H.~Modares, F.~L. Lewis, and Z.-P. Jiang, ``Optimal output-feedback control of
  unknown continuous-time linear systems using off-policy reinforcement
  learning,'' \emph{IEEE Trans. Cybern.}, vol.~46, no.~11, pp. 2401--2410, Sep.
  2016.  \url{http://ieeexplore.ieee.org/stamp/stamp.jsp?tp=&arnumber=7574391}
\BIBentrySTDinterwordspacing

\bibitem{SCC.Kamalapurkar2017a}
\BIBentryALTinterwordspacing
R.~Kamalapurkar, ``Simultaneous state and parameter estimation for second-order
  nonlinear systems,'' in \emph{Proc. IEEE Conf. Decis. Control}, Melbourne,
  VIC, Australia, Dec. 2017, pp. 2164--2169.
  \url{http://ieeexplore.ieee.org/document/8263965/}
\BIBentrySTDinterwordspacing

\bibitem{SCC.Liberzon2012}
D.~Liberzon, \emph{Calculus of variations and optimal control theory: a concise
  introduction}.\hskip 1em plus 0.5em minus 0.4em\relax Princeton University
  Press, 2012.

\bibitem{SCC.Khalil2002}
H.~K. Khalil, \emph{Nonlinear systems}, 3rd~ed.\hskip 1em plus 0.5em minus
  0.4em\relax Upper Saddle River, NJ: Prentice Hall, 2002.

\bibitem{SCC.Dinh.Kamalapurkar.ea2014}
\BIBentryALTinterwordspacing
H.~T. Dinh, R.~Kamalapurkar, S.~Bhasin, and W.~E. Dixon, ``Dynamic neural
  network-based robust observers for uncertain nonlinear systems,''
  \emph{Neural Netw.}, vol.~60, pp. 44--52, Dec. 2014.
  \url{http://www.sciencedirect.com/science/article/pii/S089360801400166X}
\BIBentrySTDinterwordspacing

\bibitem{SCC.Xian.Queiroz.ea2004}
B.~Xian, M.~S. de~Queiroz, D.~M. Dawson, and M.~McIntyre, ``A discontinuous
  output feedback controller and velocity observer for nonlinear mechanical
  systems,'' \emph{Automatica}, vol.~40, no.~4, pp. 695--700, 2004.

\bibitem{SCC.Hornik.Stinchcombe.ea1990}
K.~Hornik, M.~Stinchcombe, and H.~White, ``Universal approximation of an
  unknown mapping and its derivatives using multilayer feedforward networks,''
  \emph{Neural Netw.}, vol.~3, no.~5, pp. 551--560, 1990.

\bibitem{SCC.Hornik1991}
K.~Hornik, ``Approximation capabilities of multilayer feedforward networks,''
  \emph{Neural Netw.}, vol.~4, pp. 251--257, 1991.

\bibitem{SCC.Abu-Khalaf.Lewis2005}
M.~Abu-Khalaf and F.~L. Lewis, ``Nearly optimal control laws for nonlinear
  systems with saturating actuators using a neural network {HJB} approach,''
  \emph{Automatica}, vol.~41, no.~5, pp. 779--791, 2005.

\bibitem{SCC.Nevistic.Primbs1996}
V.~Nevistic and J.~A. Primbs, ``Constrained nonlinear optimal control: a
  converse {HJB} approach,'' California Institute of Technology, Pasadena, CA
  91125, Tech. Rep. CIT-CDS 96-021, 1996.

\end{thebibliography}
\end{document}